\documentclass[letterpaper]{article} 
\usepackage[]{aaai25}  
\usepackage{times}  
\usepackage{helvet}  
\usepackage{courier}  
\usepackage[hyphens]{url}  
\usepackage{graphicx} 
\usepackage{natbib}  
\usepackage{caption} 
\usepackage{algorithm}
\usepackage[noend]{algpseudocode}
\usepackage{color}

\def\Ap{\mathsf{Ap}}

\def\sem#1{ \llbracket #1 \rrbracket }
\def\sabotage{\crossupsidedown}
\def\tuple#1{\langle #1 \rangle}
\def\set#1{\{#1\}}
\def\next{\mathsf{X}\,}
\def\until{\,\mathsf{U}\,}
\def\release{\,\mathsf{R}\,}
\def\globaly{\mathsf{G}\,}
\def\eventualy{\mathsf{F}\,}

\def\M{\mathfrak{M}}
\def\N{\mathfrak{N}}

\def\strat{\mathfrak{S}}

\def\path{\mathcal{P}}
\def\strat{\mathfrak{S}}

\def\naww#1{\langle   #1  \rangle }

\usepackage{amssymb}
\usepackage{amsmath}
\usepackage{amsthm}

\usepackage{mathrsfs}
\newtheorem{definition}{Definition}
\newtheorem{proposition}{Proposition}
\newtheorem{thm}{Theorem}

\newtheorem{lemma}{Lemma}

\usepackage{macros}

\let\abbrv\textsf 

\def\naww#1{\langle   #1  \rangle }

\algnewcommand\algorithmicswitch{\textbf{switch}}
\algnewcommand\algorithmiccase{\textbf{case}}

\algdef{SE}[SWITCH]{Switch}{EndSwitch}[1]{\algorithmicswitch\ #1\ \algorithmicdo}{\algorithmicend\ \algorithmicswitch}%
\algdef{SE}[CASE]{Case}{EndCase}[1]{\algorithmiccase\ #1}{\algorithmicend\ \algorithmiccase}%
\algtext*{EndSwitch}%
\algtext*{EndCase}%

\usepackage{newfloat}
\usepackage{listings}
\DeclareCaptionStyle{ruled}{labelfont=normalfont,labelsep=colon,strut=off} 
\floatstyle{ruled}
\newfloat{listing}{tb}{lst}{}
\floatname{listing}{Listing}
\title{Probabilistic Obstruction Temporal Logic: \\ a Probabilistic Logic to Reason about Dynamic Models}
\author{
    Jean Leneutre,
    Vadim Malvone,
    James Ortiz
}
\affiliations{
    LTCI, Télécom Paris, Institut Polytechnique de Paris \\ Palaiseau, France

}

\usepackage{bibentry}

\begin{document}

\maketitle

\begin{abstract}
In this paper, we propose a novel formalism called Probabilistic Obstruction Temporal Logic (\abbrv{POTL}), which extends Obstruction Logic (\abbrv{OL}) by incorporating probabilistic elements. 
\abbrv{POTL} provides a robust framework for reasoning about the probabilistic behaviors and strategic interactions between attackers and defenders
in environments where probabilistic events influence outcomes. 
We explore the model checking complexity of \abbrv{POTL} and demonstrate that it is not higher than that of Probabilistic Computation Tree Logic (\abbrv{PCTL}), making it both expressive and computationally feasible for cybersecurity and privacy applications. 
\end{abstract}

%
\section{Introduction}

Understanding and quantifying uncertainty is essential in cybersecurity, and probability theory offers a robust framework for this purpose, making it particularly valuable for risk analysis. As digital systems grow increasingly complex and dynamic, effectively assessing and managing risks becomes more challenging. Probability theory allows organizations to model the likelihood of various cyber threats, such as hacking attempts, data breaches, and software vulnerabilities, which are inherently uncertain and variable. Cybersecurity professionals can estimate the likelihood of these threats materializing and assess their potential impact on systems by applying probabilistic and non-probabilistic formalisms.

Researchers have developed various solutions over the past fifty years, with formal methods emerging as a notable success. These techniques allow for the verification of system correctness by checking if a mathematical model meets the formalized desired behavior. Notably, traditional formal approaches like model checking \cite{BaierKatoen08}, initially designed for monolithic systems, have been effectively adapted to manage open and Multi-Agent Systems (\abbrv{MAS}). In recent years, the study of \abbrv{MAS} has garnered significant attention due to its wide-ranging applications in fields such as cybersecurity, robotics, and distributed computing. \abbrv{MAS} consists of two or more interacting agents, each capable of making autonomous decisions. These systems often operate in dynamic and uncertain environments, necessitating robust formal verification techniques to ensure their reliability and correctness. 

An important logic in the context of \abbrv{MAS} is Alternating-time Temporal Logic (\abbrv{ATL}) \cite{AHK02}. The latter extends \abbrv{CTL} \cite{CE81} by introducing strategic modalities, enabling the specification of properties that involve the strategic abilities of agents. \abbrv{ATL} can express whether a group of agents can achieve a certain goal regardless of the actions of other agents, making it a powerful tool for reasoning about cooperation and competition in \abbrv{MAS}. 

Another relevant formalism in this area is Obstruction Logic (\abbrv{OL}) \cite{CLM23}, which focuses on obstructions in two-player games. In \abbrv{OL}, one player, called the Demon, can temporarily disable edges in the graph as long as their total weight remains below a specified natural number, thereby preventing the other agent from achieving its temporal goal. As illustrated in their paper, \abbrv{OL} can be well-suited for representing cybersecurity problems, where a defender can activate defense mechanisms (by disabling edges) and an attacker aims to access private resources through a sequence of atomic attacks. 

In this context, a key aspect when performing cybersecurity risk analysis is to assess the \textit{likelihood} (or probability) of success of the attack scenarios. However, \abbrv{OL} did not address this aspect, where no probabilistic concepts were introduced. For the above reasons, in this paper, we present Probabilistic Obstruction Temporal Logic (\abbrv{POTL}), a logic that extends \abbrv{OL} into a probabilistic context. \abbrv{POTL} offers a comprehensive framework for analyzing the probabilistic behaviors and strategic interactions between attackers and defenders in scenarios where probabilistic events influence outcomes. We investigate the model checking complexity of \abbrv{POTL} and show that it is comparable to that of Probabilistic Computation Tree Logic, ensuring that \abbrv{POTL} remains both expressive and computationally practical for cybersecurity and privacy applications.
\paragraph{Structure of the work.} The contribution is structured as follows. Theoretical background is presented in Section~\ref{sec:back}. In Section~\ref{sec:SyntaxandSemantics}, we present the syntax and the semantics of our new logic, called Probabilistic  Obstruction Temporal Logic (\abbrv{POTL}). In Section~\ref{sec:mc}, we show our model checking  algorithm and prove that the model checking problem for \abbrv{POTL} is decidable in polyonimal-time.  
In section~\ref{sec:example}, we present an illustrative example related to the cybersecurity analysis. In Section~\ref{sec:relwork}, we compare our approach to related work. Finally, Section~\ref{sec:end}  concludes and presents possible future directions.

\section{Background}
\label{sec:back}

In this section, we discuss the basic notions that are used in the technical part of this paper.  

\paragraph{General Concepts.} 
Let $\mathbb{N}$ be the set of natural numbers, we refer to the set of natural numbers containing $0$ as $\mathbb{N}_{\ge 0}$, $\mathbb{R}_{\ge 0}$ the set of non-negative reals and $\mathbb{Z}$ the set of integers.  
Let $X$ and $Y$ be two sets and $|X|$ denotes its cardinality. The set operations of intersection, union, complementation, set difference, and Cartesian product are denoted $X$ $\cap$ $Y$, $X$ $\cup$ $Y$, $\overline{X}$, $X$ $\setminus$ $Y$, and $X$ $\times$ $Y$, respectively. Inclusion and strict inclusion are denoted $X$ $\subseteq$ $Y$ and $X$ $\subset$ $Y$, respectively. The empty set is denoted $\emptyset$. Let $\pi=x_1,\ldots,x_n$ be a finite sequence, $last(\pi)$ denotes the last element $x_n$ of $\pi$. 

\paragraph{Probability Distribution and Space.} Let $Q$ be a finite set and $\mu : Q \to [0, 1]$ be a probability distribution function over $Q$ such that $\sum_{q \in Q}$ $\mu(q) = 1$. We denote by $\mathcal{D}(Q)$ the set of all such distributions over $Q$. For a given $\mu$ $\in$ $\mathcal{D}(Q)$, $supp(\mu)$ = $\{q \in Q \mid \mu(q) > 0\}$ is called the support of $\mu$. 
The standard notation of a probability space is a triple $(\Omega, \mathcal{F}, \abbrv{Pr})$, where $\Omega$ is a sample space that represents all possible outcomes, $\mathcal{F} \subseteq \mathcal{P}(\Omega) = 2^\Omega$ is a $\sigma$-algebra over $\Omega$, i.e., it includes the empty subset, and it is closed under countable unions and complement, and $\abbrv{Pr}$: $\mathcal{F} \to [0, 1]$ is a probability measure over $(\Omega, \mathcal{F})$. We denote the set of all finite and infinite sequences of elements of $Q$ by $Q^{+}$ and $Q^{*}$, respectively.

\paragraph{Attack Graphs and Moving Target Defense Mechanisms.} 



A malicious attack is defined as an attempt by an attacker to gain unauthorized access to resources or compromise the integrity of the system. In this context, the Attack Graph (\abbrv{AG}) \cite{KK2016} is a widely recognized and increasingly popular attack model. By leveraging an AG, it is possible to model interactions between an attacker and a defender who dynamically deploys Moving Target Defense (\abbrv{MTD}) mechanisms~\cite{Cho2020}.
\abbrv{MTD} mechanisms, such as Address Space Layout Randomization (\abbrv{ASLR})~\cite{app9142928}, are active defenses that use partial system reconfiguration to alter the attack surface and reduce the chances of success of the attack. However, activating an MTD countermeasure impacts system performance: during reconfiguration, system services may be partially or completely unavailable. Thus, it is crucial to select MTD deployment strategies that minimize both residual cybersecurity risks and the negative impact on system performance. However, despite the progress made in the field of \abbrv{AG} \cite{KK2016,CattaLM23} none of them takes into account some of the uncertainties in the network. Probabilistic Attack Graphs (\abbrv{PAG}) are \abbrv{AG} enriched with 
probabilities that model the likelihood of compromise of each node in the graph based on their specific characteristics \cite{LHJ22, MSK23}. 


\paragraph{Kripke Structure and Markov Chain. }

A \abbrv{PAG} can be viewed as a Probabilistic Kripke Structure (\abbrv{PKS}).
Now, we will formally define \abbrv{PKS}, the Kripke structure that is used to represent all the possible attacks on a networked system.

\begin{definition}[\abbrv{Kripke Structure}] A Kripke Structure (\abbrv{KS}) over a set $\Ap$ of atomic propositions is a tuple $\mathcal{K}$ = $\tuple{Q, q_0, R, \mathcal{L}}$ where $Q$ is a  finite, non-empty set of states, $q_0$ $\in$ $Q$ is the initial state, $R\subseteq Q\times Q$ is a binary serial relation over $Q$ (i.e., for any $q\in Q$ there is a $q'\in Q$ such that $\tuple{q,q'}\in R$) and  $\mathcal{L}:Q \to 2^{\Ap}$ is a labeling function assigning a set of atomic propositions to any state $q\in Q$. 
\end{definition}
    
\begin{definition}[\abbrv{Markov Chain}] A Markov Chain (\abbrv{MC}) is a pair $\mathcal{H}$ = $(Q, \mathbf{P})$ where $Q$ is a (countable) set of states and $\mathbf{P}$: $Q \times Q$ $\to$ $[0, 1]$ is a transition probability function such that for all state $q \in Q$,  $\Sigma_{q' \in Q} \mathbf{P}(q, q') = 1$. If $Q$ is finite, we can consider $\mathbf{P}$ to be a transition matrix. 
\end{definition}
A \abbrv{KS} can be extended via \abbrv{MC} \cite{KLS2012} to define \emph{Probabilistic Kripke Structure (PKS)} as follows.


\begin{definition}[Probabilistic Kripke Structure] A \abbrv{PKS} over a set $\Ap$ of atomic propositions is a tuple $\mathcal{G}$ = $\tuple{Q, q_0, \mathbf{P}, \mathcal{L}}$ where $(Q, \mathbf{P})$ is a \abbrv{MC}, $q_0$ is the initial state and $\mathcal{L}:S \to 2^{\Ap}$ is a labeling function assigning a set of atomic propositions to any state $q\in Q$. 
\end{definition}

\paragraph{Path.} A path $\pi$ over $\mathcal{G}$ is a finite or infinite sequence of states $\pi$ = $q_0, q_1, q_2, \ldots$  starting in the initial state $q_0$ that are built by consecutive steps, i.e.,  $\mathbf{P}(q_i,q_{i+1}) > 0$ for all $i\in \mathbb{N}$. We write $\pi_i$ to denote the $i$-th element $q_i$  of $\pi$, $\pi_{\leq i}$ to denote the prefix $q_0,\ldots, q_i$ of $\pi$, and $\pi_{\geq i}$ to denote the suffix $q_i,q_{i+1}\ldots$ of $\pi$. The set of all finite paths starting from $q \in Q$ in the model $\mathcal{G}$ is denoted by $\abbrv{Paths}_{\mathcal{G}, q}^{+}$, and the set of all infinite paths starting from $q$ is denoted by $\abbrv{Paths}_{\mathcal{G}, q}^{*}$. A \emph{history} $h$ is any finite prefix of some path. We use $H$ to denote the set of histories. Write $last(h)$ for the last state of a history $h$.

\paragraph{Cylinder.} 

We need to measure the probability of certain sets of paths. Formally, to every $q$ $\in$ $Q$ we associate the \abbrv{probability space} $(\Omega, \mathcal{F}, \abbrv{Pr})$  where $\mathcal{F}$ is the $\sigma$-algebra generated by all basic cylinders 
sets of paths called cylinder sets, which gather all paths sharing a given finite prefix (i.e., $\abbrv{Prefix}(\pi)$). Given a finite path $\hat{\pi}$ = $q_0,q_1, \ldots, q_n$ of states, the cylinder set of $\hat{\pi}$, denoted $\abbrv{Cyl}(\hat{\pi}$) = $\{\pi \in \abbrv{Paths}_{\mathcal{G}, q_0}^{*}  \mid  \hat{\pi} \in \abbrv{Prefix}(\pi) \}$, is the set of infinite paths $\pi = q_0,q_1, \cdots, q_n, \cdots$, where $\hat{\pi}$ is a prefix of $\pi$. The set of infinite paths is supposed to be equipped with the $\sigma$-algebra generated by the cylinder sets of the finite paths and the probability measure given by $\abbrv{Pr}_{\mathcal{G}}^{q_0}$$(\abbrv{Cyl}(\hat{\pi}))$ = $\prod_{i=0}^{n-1} \mathbf{P}(q_i, q_{i+1})$. The extension of $\abbrv{Pr}_{\mathcal{G}}^{q}$ from cylinders to the $\sigma$-algebra they generate is unique, and we still denote it $\abbrv{Pr}_{\mathcal{G}}^{q}$. Note that not all sets of paths are measurable with respect to $\abbrv{Pr}_{\mathcal{G}}^{q}$, but the sets we will consider in this paper are simple enough to avoid such difficulties. 
For the mathematical details of the underlying $\sigma$-algebra and probability measure refer to \cite{BaierKatoen08}. 

\paragraph{Predecessors and Successors.} 
Let  $\mathcal{G}$ be a \abbrv{PKS} and $q$ $\in$ $Q$ be one of its states, $\abbrv{pre}(q)$ denotes the set of predecessors of $q$, i.e., $\abbrv{pre}(q)=\{q'\in Q  \mid  \mathbf{P}(q', q) > 0 \}$. Similarly, $\abbrv{post}(q)$  denotes the set of successors  of $q$, i.e., $\abbrv{post}(q)= \{q'\in Q \mid  \mathbf{P}(q, q') > 0\}$, and $\abbrv{E}(q)$ denotes its outgoing edges $\abbrv{E}(q) = \{e \in Q \times Q \mid e=(q, q') \text{ for some $q'\in Q$ and }  \mathbf{P}(q, q') > 0 \}$.

\section{Model and Logic}\label{sec:SyntaxandSemantics}

In this section, we define the syntax and semantics of our Probabilistic Obstruction Temporal Logic (\abbrv{POTL}). 
To do this, first, we introduce the Probabilistic Obstruction Temporal Structure (\abbrv{POTS}), the type of model that we use to verify \abbrv{POTL} properties.

\begin{definition}[Probabilistic Obstruction Temporal Structure] A \abbrv{POTS}  (model for short) is given by a tuple $\mathcal{M}$ = $(Q, q_0, \mathbf{P}, \mathcal{L}, \mathsf{C})$ where $\mathcal{G}$ = $(Q, q_0, \mathbf{P},\mathcal{L})$ is a \abbrv{PKS} and $\mathsf{C}: Q \times Q \to \mathbb{N}$ is a function assigning to any  pairs $(q,q')$  a natural number $n \in \mathbb{N}$. 
\end{definition}

\paragraph{Strategy and Outcomes.}
Let $\mathcal{M}$ be a model, $Q$ be states in $\mathcal{M}$, $\mathsf{C}$ is the function cost and $n$ be a natural number, a $n$-strategy is a function $\strat: H\to 2^{Q \times Q}$ that, given a history $h$, returns a subset $T \in Q \times Q$, such that: 
    (i) $T\subset \abbrv{E}(last(h))$, 
    (ii) $(\sum_{e\in T} \mathsf{C}(e)) \leq n$. 
 A memoryless n-strategy is a n-strategy $\strat$ such that for all histories $h$ and $h'$ if $last(h)=last(h')$ then $\strat(h)=\strat(h')$. A memoryless n-strategy can be seen as a function whose domain is the set $Q$ of states of a model $\mathcal{M}$.  As in \abbrv{ATL} logic, the notion of a path that is compatible with a strategy is central to the semantics of Probabilistic Obstruction Logic (\abbrv{POTL}) formulas. We define this notion by saying that a path $\pi$ is compatible with an n-strategy $\strat$ if for all $i\geq 1$ we have that $(\pi_i,\pi_{i+1})\notin \strat(\pi_{\leq i})$.  The set of outcomes of an $n$-strategy  $\strat$  and state $q$ is denoted as $\abbrv{Out}(q, \strat)$ and it returns the set of all paths that can result from a strategy $\strat$ and a state $q$.
As said in the introduction, our logic (\abbrv{POTL}) aims to capture strategies for a particular type of game played over a \abbrv{POTS}, in such games, one of the two players (the Demon) has the power to temporally deactivate some transitions of the model. We now introduce the syntax of our logic. 

\begin{definition}

Let $\Ap$ be an at most countable set of atomic formulas (or atoms). Formulas of Probabilistic  Obstruction Temporal Logic (\abbrv{POTL}, for short) are defined by the following grammar:

$$\varphi ::= \top \mid p \mid \neg \varphi \mid \varphi \land \varphi \mid \langle \sabotage_n^{\bowtie k} \rangle  \theta$$
$$\theta::= \prox \varphi \mid  \varphi \until^{\le m} \varphi \mid \varphi \until \varphi \mid \varphi \release^{\le m} \varphi \mid \varphi \release \varphi $$

where $p$ $\in$ $\Ap$ is an atomic formula, $k$ $\in$ $[0, 1]$ is a rational constant, $n$ (the grade) and $m$ are any natural number in $\mathbb{N}$, and $\bowtie$ $\in$ $\{\le, <, >, \ge\}$.
\end{definition}

In the above syntax, we distinguish between state formulas $\varphi$ and path formulas $\theta$. State formulas are evaluated over states and path formulas over paths. A model property is always expressed as a state formula, path formulas appear only as parameters of the probabilistic path operator $\langle \sabotage_n^{\bowtie k} \rangle  \theta$. 
The operators $X$ (next), $\until^{\le m}$ (bounded until), $\until$ (until), $\release^{\le m}$ (bounded release), and $\release$ (release), which are standard in temporal logic, are allowed as path formulas. The number $n$ is called \emph{the grade} of the strategic operator. The boolean connectives $\bot$, $\vee$ and $\to$ can be defined as usual, we define $\langle {\sabotage_n^{\bowtie k}} \rangle \mathsf{F} \varphi:= \langle {\sabotage_n^{\bowtie k}} \rangle (\top \until \varphi)$, $\langle {\sabotage_n^{\bowtie k}} \rangle \globaly \varphi:= \langle {\sabotage_n^{\bowtie k}} \rangle ( \bot \release \varphi) $ and $\langle {\sabotage_n^{\bowtie k}}  \rangle (\varphi\, \mathsf{W}\, \psi):= \langle {\sabotage_n^{\bowtie k}} \rangle (\psi \release (\varphi \vee \psi))$.  
The size $|\varphi|$ of a formula $\varphi$ is the number of its connectives. The intuitive meaning of a formula $\langle {\sabotage_n^{\bowtie k}} \rangle \varphi$ with $\varphi$ temporal formula is “there is a demonic strategy such that all paths of the graphs that are compatible with the strategy satisfy $\varphi$ with a probability in relation $\bowtie$ with constant $k$” where “demonic strategy” means “a strategy for disabling arcs”.  Formulas of \abbrv{POTL} can be interpreted over \abbrv{POTS}. We can now precisely define the semantics of \abbrv{POTL} formulas. 

%


\begin{definition}\label{def:sat}
    The satisfaction relation between a model $\mathcal{M}$, a state $q$ of $\mathcal{M}$, and a formula $\varphi$ is defined by induction on the structure of $\varphi$: 
\begin{itemize}
    \item $\mathcal{M}, q\models \top$ for all state $q$,
    \item  $\mathcal{M}, q\models p$ iff $p\in \mathcal{L}(q)$,
    \item  $\mathcal{M}, q\models \neg \varphi$ iff not $\mathcal{M}, 
    q\models \varphi$ (notation $\mathcal{M}, q\not \models \varphi$),
    \item $\mathcal{M}, q\models \varphi_1 \land \varphi_2$ iff $\mathcal{M}, q\models \varphi_1$ and $\mathcal{M}, q\models \varphi_2$, 
    \item $\mathcal{M}, q\models \langle \sabotage_n^{\bowtie k} \rangle  \theta $  iff there is a n-strategy $\strat$ such that $\abbrv{Pr}_{\mathcal{M}}^{q}$$(\{\pi \in  \abbrv{Out}(q,\strat) \ \mid \ \mathcal{M}, \pi \models \theta \})$ $\bowtie$ $k$.
     \end{itemize}

\noindent The satisfaction relation  $\mathcal{M}, \pi \models \varphi$ between a model $\mathcal{M}$, a path $\pi$ $\in$ $\abbrv{Paths}_{\mathcal{M}, q}^{*}$ of $\mathcal{M}$,  and path formula $\theta$ is defined as follows:

\begin{itemize}
    \item $\mathcal{M}, \pi \models \prox \varphi$ iff $\mathcal{M}, \pi_2 \models \varphi$,
     \item $\mathcal{M}, \pi \models \varphi_1 \until^{\le m} \varphi_2$ iff there is an $ 0 \le i \le m$ such that $\mathcal{M},\pi_i \models \varphi_2$ and $\mathcal{M},\pi_j \models \varphi_1$ for all $0 \leq j < i$,
    \item $\mathcal{M}, \pi \models \varphi_1 \until \varphi_2$ iff there is an $i\geq 0$ such that $\mathcal{M},\pi_i \models \varphi_2$ and $\mathcal{M},\pi_j \models \varphi_1$ for all $0\leq j < i$,
    \item $\mathcal{M},\pi \models \varphi_1 \release^{\le m} \varphi_2$ iff either $\mathcal{M}, \pi_i \models \varphi_2 $ for all $0 \le i \le m$ or there is an $0 \le i \le m$ such that $\mathcal{M},\pi_i \models \varphi_1$ and $\mathcal{M}, \pi_j \models \varphi_2$ for all $0 \leq j \leq i$.
    \item $\mathcal{M},\pi \models \varphi_1 \release \varphi_2$ iff either $\mathcal{M}, \pi_i \models \varphi_2 $ for all $i\geq 0$ or there is an $i \geq 0$ such that $\mathcal{M},\pi_i \models \varphi_1$ and $\mathcal{M}, \pi_j \models \varphi_2$ for all $0 \leq j \leq i$.
\end{itemize}

\end{definition}
Let $\varphi$ be a formula and $\mathcal{M}$ be a model,  then \abbrv{Sat}$({\varphi}, \mathcal{M})$ denotes the set of states of $\mathcal{M}$  verifying $\varphi$, i.e., \abbrv{Sat}$(\varphi, \mathcal{M})=\set{q \in Q \ \, | \, \ \mathcal{M}, q \models \varphi}$. 
%
%
Two formulas $\varphi$ and $\psi$ are equivalent (denoted by $\varphi \equiv \psi$) if for all models $\mathcal{M}$, \abbrv{Sat}$(\varphi,\mathcal{M})$ = \abbrv{Sat}$(\psi,\mathcal{M})$
The semantics of the obstruction probabilistic operator $\langle  {\sabotage_n}^{\bowtie k} \rangle$ refers to the probability for the sets of paths for which a path formula holds. To ensure that this is well-defined, we need to establish that the events specified by \abbrv{POTL} path formulas are measurable. Since the set $\{ \pi \in \abbrv{Out}(q,\strat) \mid \mathcal{M}, \pi \models \varphi \}$ for \abbrv{POTL} path formula $\varphi$ can be considered as a countable union of cylinder sets, its measurability is ensured. This follows from the following lemma.

\begin{lemma}
    For each $\abbrv{POTL}$ path formula $\varphi$ and state $q$ of a model $\mathcal{M}$, the set $\{ \pi \in  \abbrv{Out}(q,\strat) | \ \mathcal{M}, \pi \models \varphi\}$ is measurable.
\end{lemma}
\begin{proof}
   The approach is similar to the one proposed in \cite{BaierKatoen08} for \abbrv{PCTL}. 
\end{proof}






\section{Model Checking}\label{sec:mc}

Here, we present our model checking algorithm for \abbrv{POTL}. Furthermore, we show that the model checking problem for \abbrv{POTL} is decidable in \abbrv{PTIME}.  \abbrv{POTL} model checking algorithm is based on the computation of the set \abbrv{Sat}$(\varphi, \mathcal{M})$ of all states satisfying a \abbrv{POTL} formula $\varphi$, followed by checking whether the initial state is included in this set.
The most interesting part of our logic is the treatment of the formula  $\psi$ = $\naww{\sabotage_n^{\bowtie k}}  \theta$.
 In order to determine whether $q$ $\in$ $\abbrv{Sat}(\psi, \mathcal{M})$,  we will use $\abbrv{Pr}_{\mathcal{M}, q}^{\strat}$ $(\theta)$ to denote the probability that all paths from $q$ that are in accordance with the n-strategy $\strat$ satisfies path formula $\theta$, that is $\abbrv{Pr}_{\mathcal{M}, q}^{\strat}(\theta) = \abbrv{Pr}_{\mathcal{M}}^{q}(\{\pi \in  \abbrv{Out}(q,\strat)  \mid \mathcal{M}, \pi \models \theta \})$. Then  
 $$\abbrv{Sat}(\psi, \mathcal{M}) = \{q \in Q \ \mid \  \abbrv{Pr}_{\mathcal{M}, q}^{\strat}(\theta) \bowtie k \}$$

We omit the superscript $\mathcal{M}$ in $\abbrv{Pr}_{q}^{\strat}$ $(\theta)$  and $\abbrv{Sat}(\psi)$ when the model is clear from the context. 
Now, we introduce our predecessor operator. Let $A$ = $\abbrv{Sat}(\varphi)$ be a set of states, then the predecessor computation is done by the obstruction predecessor operator  $\blacktriangledown(n, A)$ where $n$ is an integer and the operator computes the set of all predecessor states. 

\begin{definition}
\label{def:definition13}
  Given a set of  states  $A$ $\subseteq$ $Q$, we define $\abbrv{Pre}(A)$ = $ \bigcup_{q \in  A}$ $ \abbrv{pre}(q)$.   
\end{definition}

Now, let us define the obstruction predecessor operator. 


\begin{definition}[Obstruction Predecessor]
  Let $\mathcal{M}=(Q, q_0, \mathbf{P}, \mathcal{L}, \mathsf{C})$ be a model. 
Given a state  $q$,  a natural number $n$, and  a set of states $A$ $\subseteq$ $Q$, we write:

$$\blacktriangleright(q,n,A) \ = \ \sum_{q'\in A} \mathsf{C}(q,q')\leq n $$ 
$$\blacktriangledown(n,A) =  \set{q \in \abbrv{Pre}(A) \, \mid \,  \blacktriangleright(q,n,\overline{A})}$$

\end{definition}

The general structure of the  Algorithm~\ref{alg:labeling} shown here is similar to \abbrv{OL} model checking algorithm \cite{CLM23}. However, it is now necessary to compute relevant probabilities. For model checking operator  $\naww{\sabotage_n^{\bowtie k}} \theta$ applied to a model $\mathcal{M}$ the probability of a path leaving each state $q$ satisfying the path formula $\theta$ must be computed. This may require a calculation involving the operators: next $(\prox \varphi)$, bounded until ($\varphi_1 \until^{\le m} \varphi_2$), until ($\varphi_1 \until \varphi_2$), bounded release ($\varphi_1 \release^{\le m} \varphi_2$), or release  ($\varphi_1 \release \varphi_2$).  We calculate, for an n-strategy and all states $q$ $\in$ $Q$, the probabilities: $\abbrv{Pr}_{q}^{\strat}(\prox \varphi)$, $\abbrv{Pr}_{q}^{\strat}(\varphi_1 \until^{\le m} \varphi_2)$, $\abbrv{Pr}_{q}^{\strat} (\varphi_1 \until \varphi_2)$, $\abbrv{Pr}_{q}^{\strat}(\varphi_1 \release^{\le m} \varphi_2)$, and $\abbrv{Pr}_{q}^{\strat} (\varphi_1 \release \varphi_2)$ respectively.


  \begin{algorithm}[ht]
		\caption{\abbrv{POTL} model checking  \\ \abbrv{Input:} A  model $\mathcal{M}$ and  $\varphi$  is a \abbrv{POTL} formula\\
         \abbrv{Output:} \abbrv{Sat}$(\varphi)$ $\gets$ $\{q \in Q \ | \ \mathcal{M}, q \models  \varphi \}$ }
		   \label{alg:labeling}
		    \begin{algorithmic}[1]
             \ForAll{$\psi$  $\in Sub(\varphi) $}
			\Switch{$(\psi)$}
			\Case{$\psi=\top$}
			\State{$ \abbrv{Sat}(\psi) \gets Q$}
			\EndCase
			\Case{$\psi=p$}
			\State{$  \abbrv{Sat}(\psi) \gets  \set{q\in Q \, | \,  p\in \mathcal{L}(q)}$}
			\EndCase
			\Case{$\psi=\neg \psi_1$}
			\State{$\abbrv{Sat}(\psi) \gets Q \setminus \abbrv{Sat}(\psi) $}
			\EndCase
			\Case{$\psi=\psi_1 \land \psi_2$}
			\State{$\abbrv{Sat}(\psi) \gets \abbrv{Sat}(\psi_1)\cap \abbrv{Sat}(\psi_2)$}
			\EndCase
			\Case{$\psi=\naww{\sabotage_n^{\bowtie k}} \theta$}
			\State{$\abbrv{Sat}(\psi) \gets \{q \in Q \ \mid \  \abbrv{Pr}_{\mathcal{M}, q}^{\strat}(\theta) \bowtie k \}$}
			\EndCase
			\EndSwitch
			\EndFor
		\end{algorithmic}
	\end{algorithm}

    Let us first consider the next operator. For $\psi$ = $\prox \varphi$, the following equality holds: 
 $\abbrv{Pr}_{q}^{\strat}(\prox \varphi)   = \min_{q' \in \blacktriangledown(n, \abbrv{Sat}(\varphi))} \mathbf{P}(q', q)$, where $\mathbf{P}$ is the transition probability function of $\mathcal{M}$. Thus, we have the resulting vector $(\abbrv{Pr}_{q}^{\strat}(\prox \varphi))_{q \in \abbrv{Sat}(\varphi)}$. 
 
    Let us consider the bounded until operator. For $\psi$ = $\varphi_1 \until^{\le m} \varphi_2$, the following equality holds: $\abbrv{Pr}_{q}^{\strat} (\varphi_1 \until^{\le m} \varphi_2)$. The set of states is partitioned into the three disjoint sets to perform the computation associated with this operator: $Q^{no}$ = $Q \setminus (\blacktriangledown(n, \abbrv{Sat}(\varphi_1))$ $\cup$ $\blacktriangledown(n, \abbrv{Sat}(\varphi_2))$, $Q^{yes}$ = $\blacktriangledown(n, \abbrv{Sat}(\varphi_2))$, and $Q^{?}$ = $Q \setminus (Q^{no} \cup Q^{yes})$. The sets $Q^{yes}$ and $Q^{no}$ contain the states for which $\abbrv{Pr}_{q}^{\strat} (\varphi_1 \until^{\le m} \varphi_2)$ is equal to 1 and 0 respectively, and $Q^{?}$ contains all other states. For the set of states $Q^{?}$ we have:

 \begin{equation*}
     \mathcal{X}_{q}^{m}= 
     \begin{cases}
         0  &  \text{if $m=0$} \\
        \min_{q' \in \blacktriangledown(n, \abbrv{Sat}(\varphi_1))} \mathbf{P}(q', q) \cdot 
      \mathcal{X}_{q}^{m-1} &
         \text{if $m \ge 1$}  
     \end{cases}
 \end{equation*}
where $\mathcal{X}_{q}^{m}$ = $ \abbrv{Pr}_{q}^{\strat}(\varphi_1\until^{\le m}\varphi_2$).  This is essentially applying the next operator $m$ times, while checking the satisfaction of $\varphi_1$ and $\varphi_2$. 
Let  $(\mathcal{X}_{q}^{m})$ = $(\abbrv{Pr}_{q}^{\strat}(\varphi_1 \until^{\le m} \varphi_2)$ be a state indexed vector and by defining the matrix $\mathbf{P}'$ as follows:

 \begin{equation*}
  \mathbf{P}'(q', q) = 
     \begin{cases}
        \mathbf{P}'(q', q)   &  \text{if $q \in Q^{?}$} \\
        $1$  & \text{if $q \in Q^{yes}$ and $q' =q$}  \\
         $0$ &   \text{if $q \in Q^{no}$}
        
     \end{cases}
 \end{equation*}
 The probabilities can be computed as follows. If $m = 0$ and $q$ $\in$ $Q^{yes}$,  then $(\mathcal{X}_{q}^{0})$ = 1, and if $q \in Q^{no}$, $(\mathcal{X}_{q}^{0})$ = 0. In the case where $m \ge  1$, the vector $(\mathcal{X}_{q}^{m})$ can be computed by $m$ matrix-vector multiplication $(\mathcal{X}_{q}^{m})$ = $\mathbf{P}'$ $\cdot$ $(\mathcal{X}_{q}^{m-1})$.
    
   Now consider the (unbounded) until operator. For $\psi$ = $\varphi_1 \until \varphi_2$ the following equality holds: $\abbrv{Pr}_{q}^{\strat} (\varphi_1 \until \varphi_2)$ \textcolor{blue}. As with the bounded until operator, all states are partitioned into the three disjoint sets $Q^{yes}$, $Q^{no}$, and $Q^{?}$. The sets are defined as above. However, the sets $Q^{yes}$, $Q^{no}$ are extended to contain all states for which $\abbrv{Pr}_{q}^{\strat} (\varphi_1 \until \varphi_2)$ is 1 or 0.  They can be determined with the fixed-point algorithms described in Algorithm~\ref{alg:labeling1} (\abbrv{Algo2}) and Algorithm~\ref{alg:labeling2} (\abbrv{Algo3}), respectively. \abbrv{Algo2} (i.e., the set $Q^{no}$) is computed by first computing the set of states reachable with non-zero probability that satisfy $\varphi_2$ whose predecessors do not satisfy $\varphi_1$. Subtracting these states from the set $Q$ gives the set of states with 0 probability. \abbrv{Algo3} (i.e. the set $Q^{yes}$) computes similarly the set of states that are reachable with probability less than 1 and that satisfy $\varphi_2$ whose predecessors do not satisfy $\varphi_1$. The set of states satisfying the operator with probability 1 is determined by subtracting these states from $Q$. The reason for precomputing $Q^{yes}$, $Q^{no}$ is that it ensures a unique solution to the linear system of equations and reduces the set of states in $Q^{?}$ for which probabilities must be computed numerically. In addition, the model checking of qualitative properties for which the probability bound is 1 or 0 does not require any further computation. The final set $Q^{?}$ can be computed by solving the linear equation.

 \begin{equation*}
  \mathcal{X}_{q} =
     \begin{cases}
        $0$ &  \text{if $q$ $\in$ $Q^{no}$} \\
        $1$ &  \text{if $q \in Q^{yes}$}  \\
       \min_{q' \in Q} \mathbf{P}(q', q) \ \cdot \  \mathcal{X}_{q}'  &   \text{if $q$ $\in$ $Q^{?}$}
        
     \end{cases}
 \end{equation*}
 
 where  $\mathcal{X}_{q}$ = $\abbrv{Pr}_{q}^{\strat}(\varphi_1\until \varphi_2)$. To reconstruct the problem in the form $\mathbf{A}\cdot x = b$. Let  $(\mathcal{X}_{q})$ be the state indexed vector where $(\mathcal{X}_{q} = 1$ if $q$ $\in$ $Q^{yes}$ and $(\mathcal{X}_{q}) = 0 $ if $q$ $\in$ $Q^{no}$, and $\mathbf{A}$  = $\mathbf{I}$ – $\mathbf{P}'$ where $\mathbf{I}$ is the identity matrix and matrix $\mathbf{P}'$ is as defined below: 
\begin{equation*}
  \mathbf{P}'(q', q) = 
     \begin{cases}
        \mathbf{P}'(q', q)   &  \text{if $q \in Q^{?}$} \\
        $1$  & \text{if $q \in Q^{yes}$}  \\
         $0$ &   \text{if $q \in Q^{no}$}
        
     \end{cases}
 \end{equation*}
The Power method \cite{EVR62}, can then be used to solve the linear system $\mathbf{A}\cdot x = b$.
 

 \begin{algorithm}
		\caption{Backward search for computing  $Q^{no}$ (\abbrv{Algo2})  \\ \abbrv{Input:} A  formula $\abbrv{Pr}_{q}^{\strat}(\abbrv{Sat}(\varphi_1) \until \abbrv{Sat}(\varphi_2))$ and $Q$.  \\
         \abbrv{Output:} A set \abbrv{R}  of states which have a zero probability. }
		   \label{alg:labeling1}
		    \begin{algorithmic}[1]
			\State{$ Y \gets \abbrv{Sat}(\varphi_2)$}
            \State{$ X \gets \emptyset$}
		    \While{$(Y \neq X)$}
   			\State{$X \gets Y$}
            \State{$Y \gets Y \cup (\{q \in \abbrv{Sat}(\varphi_1)  |$ $  \exists q' \in Y, \textbf{P}(q', q)>0\}$  $\cap$ $\blacktriangledown(n, \abbrv{Sat}(X)))$} 
		\EndWhile
             \State{$R \gets Q \setminus Y$}
		    \State{$ \mathbf{return} \  R$} 
		\end{algorithmic}
	\end{algorithm}
\begin{algorithm}
		\caption{Backward search for computing   $Q^{yes}$ (\abbrv{Algo3})   \\ \abbrv{Input:} A   formula $\abbrv{Pr}_{q}^{\strat}(\abbrv{Sat}(\varphi_1) \until \abbrv{Sat}(\varphi_2))$, $Q$ and $Q^{no}$  \\
         \abbrv{Output:}  A set \abbrv{R}  of states satisfying the operator with  probability 1. }
		   \label{alg:labeling2}
		    \begin{algorithmic}[1]
			\State{$ Y \gets Q^{no}$}
            \State{$ X \gets \emptyset$}
		    \While{($Y \neq X)$}
   			\State{$X \gets Y$}
            \State{$Y \gets Y \cup (\{q \in (\abbrv{Sat}(\varphi_1) \setminus \abbrv{Sat}(\varphi_2)) |$ $ \exists q' \in Y,$ $ \textbf{P}(q',q)>0\}$ $\cap$ $\blacktriangledown(n, \abbrv{Sat}(X)))$} 
		\EndWhile
             \State{$R \gets Q \setminus Y$}
		    \State{$ \mathbf{return} \  R$} 
		\end{algorithmic}
	\end{algorithm}

  Now consider the bounded release operator. For $\psi$ = $\varphi_1 \release^{m} \varphi_2$ the following equality holds: $\abbrv{Pr}_{q}^{\strat} (\varphi_1 \release \varphi_2)$. The argument is dual to bounded until operator. 
  
  The last case is the (unbounded) release operator. For $\psi$ = $\varphi_1 \release \varphi_2$ the following equality holds:  $\abbrv{Pr}_{q}^{\strat} (\varphi_1 \release \varphi_2)$. The argument is dual to unbounded until operator. However, Algorithm~\ref{alg:labeling1} and \ref{alg:labeling2} should be modified at line 5, where the intersection operator should be changed to the union operator.  

Let us now prove the termination and correctness of the Algorithm~\ref{alg:labeling}. 


 
\begin{thm}[\textbf{Correctness}]
Let $\mathcal{M}$ be a \abbrv{POTS} model and  $\varphi$ be a \abbrv{POTL} formula. Then, $(i)$ $\abbrv{Sat}(\varphi)$ terminates and $(ii)$ $q$ $\in$ $\abbrv{Sat}(\varphi)$ iff  $\mathcal{M}, q$ $\models$ $\varphi$.
\end{thm}
 \begin{proof}(Sketch)
 Intuitively, termination is straightforward because recursive calls within $\abbrv{Sat}(\varphi)$ are always applied to strictly sub-formulas of $\varphi$.  Let us prove $(i)$ and $(ii)$ by induction over the structure of $\varphi$ that, for every $\psi$ $\in$ $\abbrv{Sub}(\varphi)$ and $q$ $\in$ $Q$  holds iff, $\mathcal{M}, q$ $\models$ $\psi$.

\noindent (\textbf{Soundness.}) For every $\psi$ $\in$ \abbrv{Sub}($\varphi$) and $q$ $\in$ $Q$, implies $\mathcal{M}, q$ $\models$ $\psi$. We prove this by induction over the structure of $\psi$ as follows. For the \textbf{base case}: If $\psi=\top$, then $\abbrv{Sat}(\top)$ = $Q$. That means $(i)$ holds immediately and $(ii)$ follows directly from the truth definition. If $\psi=p$ ($p$ $\in$ $\Ap$), then  $\abbrv{Sat}(p)$ = $\mathcal{L}(p)$ by definition. That means $(i)$ holds immediately and for $(ii)$ $q$ $\in$ $\abbrv{Sat}(\psi)$ iff $q$ $\in$ $\mathcal{L}(p)$ then $\mathcal{M}, q$ $\models$ $p$ by the truth definition. For the \textbf{induction case}: the cases of boolean combinations, if $\psi= \neg \psi$, then $\abbrv{Sat}(\psi)$ = $Q$ $\setminus$ $\abbrv{Sat}(\psi_1)$ then induction hypothesis, $\abbrv{Sat}(\psi_1)$ terminates, therefore, $(i)$ holds. For $(ii)$ $q$ $\in$ \abbrv{Sat}$(\psi)$ iff $q$ $\in$ $Q$ $\setminus$ \abbrv{Sat}$(\psi)$ then, $q$ $\notin$ \abbrv{Sat}$(\psi_1)$ which means that $\mathcal{M}, q$ $\not\models$ $\psi_1$ by the induction hypothesis then $\mathcal{M}, q$ $\models$ $\neg \psi_1$ by truth definition. If $\psi=\psi_1 \wedge \psi_2$, then \abbrv{Sat}$(\psi)$ = \abbrv{Sat}$(\psi_1)$ $\cap$ \abbrv{Sat}$(\psi_2)$. By the
induction hypothesis,  \abbrv{Sat}$(\psi_1)$ and \abbrv{Sat}$(\psi_2)$ terminate, therefore, $(i)$ holds. For $(ii)$  $q$ $\in$ \abbrv{Sat}$(\psi)$ iff $q$ $\in$ \abbrv{Sat}$(\psi_1)$ $\cap$ \abbrv{Sat}$(\psi_2)$  then, $q$ $\in$ \abbrv{Sat}$(\psi_1)$ and $q$ $\in$ \abbrv{Sat}$(\psi_1)$ which means that $\mathcal{M}, q$ $\models$ $\psi_1$ and $\mathcal{M}, q$ $\models$ $\psi_2$ by the induction hypothesis then $\mathcal{M}, q$ $\models$ $\psi_1$ $\wedge$ $\mathcal{M}, q$ $\models$ $\psi_2$ by truth definition. The induction step for the remaining obstruction operators is as follows: If $\psi$ = $\naww{\sabotage_n^{\bowtie k}} \theta$ then \abbrv{Sat}$(\psi)$ = $\{q \in Q \ \mid \  \abbrv{Pr}_{\mathcal{M}, q}^{\strat}(\theta) \bowtie k \}$ by the \abbrv{Sat} definition. To prove that $(i)$ holds, we must show $\abbrv{Pr}_{\mathcal{M}, q}^{\strat}(\theta)$ terminates. If  $\psi$ = $\prox  \varphi_1$, the computation of $\abbrv{Pr}_{\mathcal{M}, q}^{\strat}(\theta)$ terminates due to the fact that \abbrv{Sat}$(\varphi_1)$ terminates by the induction hypothesis, and the functions $\blacktriangleright(q,n,A)$ and $\blacktriangledown(n,A)$  are finite. If $\psi$ = $\varphi_1 \until^{\le m}  \varphi_2$, the computation of $\abbrv{Pr}_{\mathcal{M}, q}^{\strat}(\theta)$
terminates due to the fact that \abbrv{Sat}$(\varphi_1)$ and \abbrv{Sat}$(\varphi_2)$ terminate by the induction hypothesis. $\blacktriangleright(q,n,A)$, $\blacktriangledown(n,A)$, \abbrv{Sat}$(\varphi_1)$ and \abbrv{Sat}$(\varphi_2)$ are all finite and the solution of the corresponding linear equation systems also terminates. If $\psi$ = $\varphi_1 \until \varphi_2$, the computation of $\abbrv{Pr}_{\mathcal{M}, q}^{\strat}(\theta)$ terminates due to the fact that \abbrv{Sat}$(\varphi_1)$ and \abbrv{Sat}$(\varphi_2)$ terminate by the induction hypothesis. 
If $\psi$ = $\varphi_1 \release^{\le m}  \varphi_2$, the computation of $\abbrv{Pr}_{\mathcal{M}, q}^{\strat}(\theta)$
terminates due to the fact that \abbrv{Sat}$(\varphi_1)$ and \abbrv{Sat}$(\varphi_2)$ terminate by the induction hypothesis. The argument is symmetric to bounded until operator.
Let $R$ be the set of symbolic states of $Q$ that is returned by algorithm~\ref{alg:labeling2} at line $6$. We need to show that $R=\abbrv{Sat}(\psi_2)$ provided that $X=\abbrv{Sat}({\psi_1})$. We first show that $\abbrv{Sat}({\psi}) \subseteq Y$. Suppose that $q\in \abbrv{Sat}(\psi)$. By the definition of satisfaction, this means that there is a strategy $\strat$ such that given any $\rho=q_1,q_2,\ldots $ in $Out(q,\strat)$ and note that since the cardinality of $\mathcal{M} $ is finite, and we can suppose that $\strat$ is memoryless, we can focus on the finite prefix $q_1,\ldots q_m$ of $\rho$ in which all the $q_i$ are distinct. Let $A_i$ (for $i< |{\mathcal{M}}|)$ be the value of the variable $A$ before the first $i$-th iteration of the algorithm. We show that if $C\subseteq A_i$ then $C\subseteq A_{i+1}$. Firstly, note that $A_i\subseteq \abbrv{Sat}(\psi_1)$ for all $i$.  By definition,  we have that $A_{i+1}=\blacktriangledown(n,A_i)\cap \abbrv{Sat}(\psi_1)$, i.e., $A_{i+1}$ is computed by taking all the element of $\abbrv{Sat}(\psi)$ that have at most $n$ successors that are not in $A_i$. 
If $\psi$ = $\varphi_1 \release \varphi_2$ then the proof is similar to the above case. 

\noindent ($\textbf{Completeness}$) For every $\psi$ $\in$ \abbrv{Sub}($\varphi$) and $q$ $\in$ $Q$, we prove that $\mathcal{M}, q$ $\not\models$ $\psi$ by induction over the structure of $\psi$ as follows. For the \textbf{base case:} If $\psi=\top$ and $\psi=p$ ($p$ $\in$ $\Ap$), are obvious. For the $\textbf{induction case}$, the cases of boolean combinations, $\psi= \neg \psi$, then $\psi$ was model checked, and it was found to be true. Thus, $\mathcal{M}, q$ $\not\models$ $\psi$. For  $\psi=\psi_1 \wedge \psi_2$, then $\psi_1$ and $\psi_2$ were model checked and at least one of them was found to be false. Therefore, $\mathcal{M}, q$ $\not\models$ $\psi$.  The induction step for the remaining obstruction operators is as follows: If  $\psi$ = $\naww{\sabotage_n^{\bowtie k}} \theta$ then \abbrv{Sat}$(\psi)$ = $\{q \in Q \ \mid \  \abbrv{Pr}_{\mathcal{M}, q}^{\strat}(\theta) \bowtie k \}$ by the \abbrv{Sat} definition. The proof for $X=\abbrv{Sat}({\psi_1})$, $\psi$ = $\varphi_1 \until^{\le m}  \varphi_2$ and $\psi$ = $\varphi_1 \until  \varphi_2$ then  $\mathcal{M}, s \not\models \psi$ is similar to the above case (similar for bounded and unbounded $\release$).

      
\end{proof}

The following theorem establishes the complexity of our model checking algorithms.

\begin{thm}\label{theoModelcheclTOL}
	The model checking problem of \abbrv{POTL} on \abbrv{POTS} is \abbrv{PTIME}
\end{thm}

\begin{proof} (Sketch). Algorithm~\ref{alg:labeling} shows a procedure for model checking \abbrv{POTL}, which manipulates a set of states of $Q$. The procedure is inspired by the model checking for \abbrv{OL}~\cite{CattaLM23}, \abbrv{PCTL}~\cite{HAJ94} and \abbrv{ATL}~\cite{AHK02}. However, we use two additional procedures $\blacktriangleright$ and $\blacktriangledown$ linked to the pre-image function $\abbrv{Pre}$. In detail, our algorithm uses the following functions:
   \begin{itemize}
    \item The function $\abbrv{Sub}$ returns an ordered sequence, w.r.t. their complexities, of syntactic sub-formulas of a given formula $\varphi$.
    \item The function $\abbrv{Pre}$ is the same as for $\abbrv{OL}$~\cite{CattaLM23}.
    \item The function $\blacktriangleright(q,n,A)$ takes in input a state $q$, a natural numbers $n$, and a subset of states $A$. Such a function returns true if $(\sum_{q'\in A} \mathsf{C}( \tuple{q, q'}))<n$. If we represent the graph via an adjacent matrix,  we can calculate such function in a  linear number of steps w.r.t. the size of $A$. 
   \item The function $\blacktriangledown(n,A)$ takes in input a natural number $n$ and a subset of states $A$. The function returns the subset $A'$ of $\abbrv{Pre}(A)$, such that $\blacktriangleright(q',n,\overline{A})$ for all $q'\in A'$. The worst possible case is when $\abbrv{Pre}(A)=Q$, and  one needs to call $|Q|$-times the function $\blacktriangleright$. So, we are quadratic in $Q$, i.e. polynomial.
    \end{itemize}

 Algorithm~\ref{alg:labeling} works bottom-up on the structure of the formula, the cases of interest are for strategic formulas. 
 For $\varphi=\naww{\sabotage_n^{\bowtie m}}\prox\varphi$, the procedure calls function $\blacktriangledown(n,\abbrv{Sat}(\varphi))$ to compute the subset of set of states of $\abbrv{Pre}(\abbrv{Sat}(\varphi_1))$ that are bound to end up in satisfaction set. As regard $\varphi=\naww{\sabotage_n^{\bowtie m}} (\varphi_1 \until \varphi_2)$, the procedure computes the least fixed-point. We observe that, since it is monotone, such a fixed-point always exists.  A similar reasoning can be done for $\varphi=\naww{\sabotage_n^{\bowtie m}} \varphi_1 \until^{\le m} \varphi_2$, $\varphi=\naww{\sabotage_n^{\bowtie m}} \varphi_1 \release^{m} \varphi_2$ and $\varphi=\naww{\sabotage_n^{\bowtie m}} \varphi_1 \release \varphi_2$. From the above, our procedure runs in polynomial-time in the size of the model and formula, 
where parameter sizes are defined as follows. The size of $\mathcal{M}$, is denoted by $|\mathcal{M}|$ and  the size of a state formula $\varphi$, denoted by $|\varphi|$, is equal to the number of logical connectives and temporal operators in $\varphi$ plus the sum of the size $log(m)$ of each bounded temporal operators $\until^{\le m}$ occurring  in $\varphi$ and the function $\blacktriangledown(n,A)$.
Therefore, checking whether a model $\mathcal{M}$ satisfies formula $\varphi$, which depends on the size of $\varphi$ and is at most $O(\vert  \varphi \vert \cdot  \vert \mathcal{M} \vert \cdot m_{max})$, where $m_{max}$  is the maximal step bound that appears in a subformula $\psi_1 \until^{\le m} \psi_2$ of $\varphi$ and if $m_{max}$ = 1,  then $\varphi$ does not contain a step-bounded until operator. Termination of such procedure is guaranteed, as the state space $Q$ is finite.

\end{proof}

\section{Illustration Example}\label{sec:example}
Probability theory is well-suited for cybersecurity risk analysis because it provides a framework for understanding and quantifying uncertainty. To illustrate this, we will consider the following general cybersecurity scenario.  Let $\mathcal{G}$ be an \abbrv{AG} and we want to check if there are \abbrv{MTD} response strategies that will satisfy certain security goals.

\begin{figure}[ht]
  \centering
  \includegraphics[width=57mm, height=20mm]{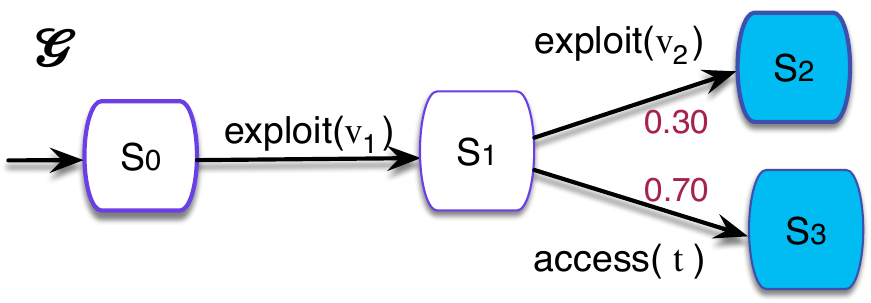}
    \vspace{-4pt}
  \caption{Example of an \abbrv{AG} $\mathcal{G}$ from~\cite{IMZ16}.}
  \label{fig:exp}
 \end{figure}
Consider the \abbrv{AG} in Fig.~\ref{fig:exp} with four states: $S_0$, $S_1$, $S_2$, and $S_3$. Each state represents a state of the attacker in the system. If the attacker is in $S_0$ or $S_1$, he can perform one or two of the following actions: exploit vulnerability $v_1$, exploit vulnerability $v_2$, and access device $t$. If the attacker succeeds in exploiting $v_1$, he will transition to state $S_1$. Here, we assume that depending on the attacker's preferences, there are 70$\%$ chance that the attacker will attempt to access equipment $t$ and a 30$\%$  chance that he will attempt to exploit $v_2$.  

\begin{table}[ht]
    \centering
    \begin{tabular}{|c|c|c|c|c|}
    \hline 
        Action & Countermeasure & Cost  & Efficiency \\
        \hline
        exploit($v_1$) & $c_1$ & 5 &  47.5$\%$\\
        access($t$) & $c_2$ & 1 &  22.5$\%$ \\
         exploit($v_2$)  & $c_3$ & 3 & 24.7$\%$ \\
         \hline
    \end{tabular}
    \vspace{-4pt}
    \caption{Actions and Attack countermeasure}
    \label{tab:my_label}
\end{table}

In Table~\ref{tab:my_label}, there are the three possible actions the attacker can deploy, with their respective countermeasures, cost, and effectiveness. Let Fig.~\ref{fig:exp1} depict the \abbrv{POTS} $\mathcal{M}$, constructed using the information from the attack graph presented in ~\cite{IMZ16}. Notice that, in contrast to ~\cite{IMZ16}, here we remove the actions because we do not have any actions in our \abbrv{POTS} model. Therefore, the probabilities present in each state of the model are divided by the number of outgoing actions of that state.
In Fig.~\ref{fig:exp1} the yellow line (do nothing), indicates that no countermeasure will be deployed. The red lines ($c_1$ in Table~\ref{tab:my_label}), refer to a defensive countermeasure aimed at protecting the system against the attack attempt. However, $c_1$ has an efficiency of 47.5$\%$. Therefore, an attacker attempting to exploit$(v_1)$ has a 5$\%$ chance of success. The violet lines ($c_2$) are a defensive countermeasure against accessing equipment $t$ and have an efficiency of 22.5$\%$. The orange lines ($c_3$) are a defensive countermeasure against exploiting vulnerability $v_2$ and have an efficiency of 24.7$\%$. Finally, green lines refer to the deployment of countermeasures $c_2$ and $c_3$ at the same time. Let us take the case where the defender chooses to deploy the countermeasure $c_3$ (orange lines) in state $S_1$, the attacker can either succeed or fail in his attack attempt. The efficiency of $c_3$ is 24.7$\%$. Therefore, the probability that the attacker fails in his attack attempt is 0.07425 (exploit$(v_2)$ $\times$ efficiency$(c_3))$. Otherwise, the probability of success is of 0.00075. 
 
\begin{figure}[ht]
  \centering
  \includegraphics[width=68mm, height=42mm]{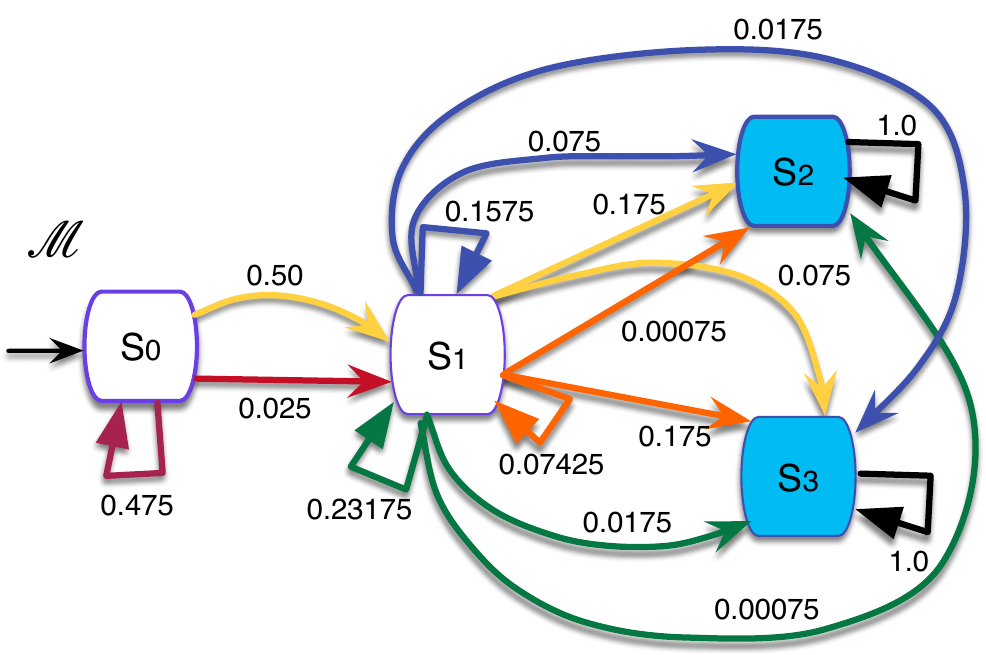}
    \vspace{-4pt}
  \caption{The \abbrv{POTS} $\mathcal{M}$ from $\mathcal{G}$.}
  \label{fig:exp1}
 \end{figure}

Let $r_2$ and $r_3$ be the atomic propositions for the states, $S_2$ and $S_3$. We can express, via \abbrv{POTL} formulas, the following security objective: 
\begin{itemize}
     \item
\textit{There
is a defender strategy with a cost $4$ such that the attacker reaches the state satisfying $r_2$ or the state satisfying $r_3$ with a probability less than a given threshold $0.1$}. \textit{The following \abbrv{POTL} formula captures the objective}: $\varphi_1:= \naww{\sabotage_{4}^{^{< 0.1}}} \eventualy  (r_2 \vee r_3)$. 
\item \textit{There exists a defender strategy with cost 5 such that the probability that the attacker reaches state satisfying $r_3$ is less than 0.2}. The following \abbrv{POTL} formula captures the objective: $\varphi_1:= \naww{\sabotage_{5}^{^{< 0.2}}}\eventualy r_3$. 

\end{itemize}
\section{Related Work}\label{sec:relwork}

There are some papers that have focused on the strategic capabilities of agents playing within dynamic game models. In this section, we compare our approach with them. 

\textbf{Non-Probabilistic Games and Strategic Logics.} Some research related to sabotage games has been introduced by van Benthem to study the computational complexity of a special class of graph reachability problems in which an agent has the ability to delete edges ~\cite{vanBenthem2005, aucherSML}. 
%
Sabotage Modal Logic (\abbrv{SML}) was introduced by~\cite{vanBenthem2005} to reason about sabotage games. The model checking problem for the sabotage modal logic is \abbrv{PSPACE-complete}~\cite{Lding2003ModelCA}. Our version of the games is not comparable to the sabotage games, because we provide the possibility to temporarily select subsets of edges, while in the sabotage games, the saboteur can only delete one edge at a time. In this respect, our work is related to~\cite{CattaLM23}, where the authors use an extended version of sabotage modal logic, called Subset Sabotage Modal Logic (\abbrv{SSML}), which allows for the deactivation of certain subsets of edges of a directed graph. The authors show that the model checking problems for such logics are decidable. 
Also, we recall that \abbrv{SSML} is an extension of \abbrv{SML}, but does not include temporal operators. Also, neither \abbrv{SML} nor \abbrv{SSML} takes into account quantitative information about the cost of edges, as we do.
In \cite{DynamicVadim} Dynamic Escape Games (\abbrv{DEG}) have been introduced.  
In a \abbrv{DEG}, an agent can inhibit edges but only reachability objectives have been studied. In \cite{CLM23} has been introduced Obstruction Logic which allows reasoning about two-player games played on weighted directed graphs. 
However, all these logics do not include quantitative information about probability and temporal operators.

%


\textbf{Probabilistic Games and Strategic Logics.}
Several papers consider the verification of stochastic games using probabilistic logics. In particular, when agents play deterministic strategies (as in \abbrv{PSL} \cite{AKM19}) and probabilistic knowledge (as in \abbrv{PATL} and $\abbrv{PATL}^{*}$ \cite{HXL19}). These logics are extensions of the Alternating-
time Temporal Logics \abbrv{ATL} and $\abbrv{ATL}^{*}$ \cite{AHK02}  and can be used to reason about the probabilistic knowledge and the probabilistic strategy in stochastic game systems.
 In \cite{SZT19}, the model checking problem has been studied for probabilistic alternating time $\mu$-calculus. \cite{HXK12} consider the logic Probabilistic $\abbrv{PATL}^{*}$ under incomplete information and synchronous perfect recall. \abbrv{PATL} has also been studied with incomplete information and memoryless strategy \cite{BJM23}, and with cumulative costs/rewards \cite{CFM12}. In the context of \abbrv{MAS}, probabilistic logic has been used to verify unconstrained parameterized systems, a fragment of $\abbrv{PATL}^{*}$ called $\abbrv{P[ATL}^{*}\abbrv{]}$ \cite{LPE20}, constrained resource systems (Probabilistic Resource-Bounded \abbrv{ ATL} (\abbrv{pRB-ATL}) \cite{NHR19}, and under assumptions about adversarial strategies, an extension of \abbrv{ATL} with probability success (\abbrv{pATL}) \cite{NWJ09}.  However, none of these logics combine probabilistic settings with dynamic models.

\section{Conclusions}\label{sec:end}

In this paper, we presented \abbrv{POTL}, a logic that allows reasoning about probabilistic two-player games temporal goals, where one of the players has the power to locally and temporarily modify the game structure. We proved that its model checking problem is in \abbrv{PTIME}. We also showed how \abbrv{POTL} expresses cybersecurity properties in a suitable way. Several directions we would like to explore for future work. A possible extension would be to consider probabilistic games with many players, between a demon and \emph{coalitions} of travelers. 
Such an extension would have the same relationship with the \abbrv{PATL} logic as \abbrv{TOTL} has with \abbrv{TCTL}. Another extension could be to introduce imperfect information in our setting. Unfortunately, this context is generally non-decidable \cite{DimaT11}. To overcome this problem, we could use an approximation to perfect information \cite{BelardinelliFM23}, a notion of bounded memory \cite{BelardinelliLMY22}, or some hybrid technique \cite{FerrandoM22,FerrandoM23}.

\bibliography{MainAAAI}

\end{document}